\documentclass[a4paper,12pt]{amsart}
\usepackage[margin=1.2in]{geometry}
\usepackage[utf8]{inputenc}
\usepackage{amsmath}
\usepackage{amsfonts}
\usepackage{amssymb}
\usepackage{graphicx}
\usepackage{epic}
\newtheorem{definition}{Definition}[section]

\newtheorem{proposition}[definition]{Proposition}

\newcommand{\greydot}{\begin{picture}(20,20)
  \put(0,0){\circle{14}}
  \put(-5,-5){\line(1,1){10}}
  \put(-5,5){\line(1,-1){10}}
  \end{picture}}
\newcommand{\vertexbn}{\begin{picture}(20,20)
  \put(6,-3){\line(1,0){30}}\put(6,3){\line(1,0){30}}
  \put(27,0){\line(-1,1){10}}\put(27,0){\line(-1,-1){10}}
  \end{picture}}
  
\newcommand{\vertexcn}{\begin{picture}(20,20)
  \put(6,-3){\line(1,0){30}}\put(6,3){\line(1,0){30}}
  \put(17,0){\line(1,-1){10}}\put(17,0){\line(1,1){10}}
  \end{picture}}
  
\newcommand{\vertexdn}{\begin{picture}(20,40)
  \put(5,5){\line(3,2){19}} \put(5,-5){\line(3,-2){19}}
  \put(31,20){\circle{14}}\put(31,-20){\circle{14}}
  \end{picture}}
\newcommand{\vertexvn}{\begin{picture}(20,40)
  \put(6,18){\line(3,-2){19}} \put(6,-18){\line(3,2){19}}
  \put(-1,20){\circle{14}}\put(-1,-20){\circle{14}}
  \end{picture}}

\begin{document}
\title{Vogan diagrams of  affine twisted Lie superalgebras }
\author{Biswajit Ransingh}
\maketitle
\address{\begin{center}
Department of Mathematics\\National Institute of Technology\\Rourkela (India)\\email- bransingh@gmail.com  
\end{center}
\date{}

\begin{abstract}
A Vogan diagram is a Dynkin diagram with a Cartan involution of twisted affine superlagebras based on maximally
compact Cartan subalgebras. This article construct the Vogan diagrams of twisted affine  superalgebras. This article is
a part of completion of classification of vogan diagrams to  superalgebras cases.
\end{abstract}
\noindent 2010 AMS Subject Classification : 17B05, 17B22, 17B40

\section{Introduction}
Recent study of denominator identity of Lie superalgebra by Kac and et.al \cite{kac:denominator} shows a number of application to number theory,
Vaccume modules and W algebras. 
We loud denominator identiy because it has  direct linked to real form of Lie superalgebra and the primary ingredient  of Vogan
diagram is classification of real forms. It follows the study of twisted affine Lie superalgebra for similar application. Hence it is essential 
to roam inside the depth on  Vogan diagram of twisted affine Lie superalgebras.

The real form of  Lie superalgebra have a wider application not only in mathematics but also in theoretical physics.
Classification of real form is always an important aspect of Lie superalgebras. There are two methods
to classify the real form one is Satake or Tits-Satake diagram other one is Vogan diagrams. The former 
is based on the technique of maximally non compact Cartan subalgebras and later is based on maximally 
compact Cartan subalgebras. The Vogan diagram first introduced by A W Knapp to classifies the real 
form of semisimple Lie algebras and  it is named after David Vogan. Since then the classification of
Vogan diagram by different authors for affine Kac-Moody algebras (untwisted and twised), hyperbolic Kac-Moody algebras ,
Lie superalgebras and affine untwisted Lie superalgebras already developed. In this article we will developed
Vogan diagrams of the rest superalgebras, twisted affine Lie superalgebras.

The classification of symmetric spaces by Satake diagram had been done in \cite{helgason:symmetric}. Similar classification of symmetric spaces had achieved 
with double Vogan diagram by Chuah in \cite{double}. Recently the study of Kac-Moody symmetric spaces obtained by Freyn \cite{Freyn}. We hope an anologus 
classification of Kac-Moody symmetric superspaces can be obtained by Vogan graph theoretical method. So the exploration
of Vogan digram of twisted Lie superalgebra is a preliminary step towards  Kac-Moody symmetric superspaces.
\section{Generalities}
The following preliminary section deals with basic structure of Lie superalgebras which is helpfull for twisted affine extension 
of  Lie superalgebras in the subsequent section.
\subsection{The general  linear Lie superalgebras}

Let $V=V_{\overline{0}}\oplus V_{\overline{1}}$ be a vector superspace,
so that End$(V)$ is an associative superalgebra. The End$(V)$ with 
the supercummutator forms a Lie superalgebra, called the general linear
Lie superalgebra and is denoted by $\mathfrak{gl}(m|n)$, where $V=\mathbb{C}^{m|n}$.
With respect to an suitable ordered basis of End$(V)$ , $\mathfrak{gl}(m|n)$
can be realized as $(m+n)\times(m+n)$ complex matrices of the block form.

$\left(\begin{array}{cc}
a & b\\
c & d
\end{array}\right)$ where $a$ , $b$, $c$ and $d$ are respectivily $m\times m$, $m\times n$,
$n\times m$ and $n\times n$ matrices. The even subalgebra of $\mathfrak{gl}(m|n)$
is $\mathfrak{gl}(m)\oplus\mathfrak{gl}(n)$ , which consists of matrices
of the form $\left(\begin{array}{cc}
a & 0\\
0 & d
\end{array}\right)$, While the odd subspace consists of 
$\left(\begin{array}{cc}
0 & b\\
c & 0
\end{array}\right)$

\begin{definition}{\em
A Lie superalgebras $\mathcal{G}$ is an algebra graded over $\mathbb{Z}_{2}$ , i.e., $\mathcal{G}$
is a direct sum of vector spaces $\mathcal{G}=\mathcal{G}_{\overline{0}}\oplus \mathcal{G}_{\overline{1}}$, 
and such that the bracket satisfies 

\begin{enumerate} 
\item $[\mathcal{G}_{i},
\mathcal{G}_{j}]\subset\mathcal{G} _{i+j(\mbox{ mod }2)}$,
\item $[x,y]=-(-1)^{|x||y|}[y,x]$, (Skew supersymmetry)
$\forall$ homogenous $x$ and $y\in \mathcal{G}$ (Super Jacobi identity)
\item $[x,[y,z]]=[[x,y],z]+\left(-1\right)^{|x||y|}[y,[x,z]]\forall z\in \mathcal{G}$
\end{enumerate} 

A bilinear form $(.,.):\mathcal{G}\times\mathcal{G}\rightarrow\mathbb{C}$ on a Lie superalgebra is 
called \textbf{invariant} if $([x,y],z)=(x,[y,z])$, for all $x,y,z\in \mathcal{G}$  

The Lie superalgebra $\mathcal{G}$ has a root space decomposition with respect to $\mathfrak{h}$ 
\begin{displaymath} \mathcal{G}=\mathfrak{h}\oplus\bigoplus_{\alpha\in\triangle}\mathcal{G}_\alpha \end{displaymath}}
\end{definition}
A root $\alpha$ is even if $\mathcal{G}_{\alpha}\subset\mathcal{G}_{\overline{0}}$
and it is odd if $\mathcal{G}_{\alpha}\subset\mathcal{G}_{\overline{1}}$ 

A \textit{Cartan subalgebra} $\mathfrak{h}$ of diagonal matrices
of $\mathcal{G}$ is defined to be a Cartan subalgebra of the even
subalgebra $\mathcal{G}_{\overline{0}}$. Since every inner automorphism
of $\mathcal{G}_{\overline{0}}$ extends to one of Lie superalgebra
$\mathfrak{g}$ and Cartan subalgebras of $\mathcal{G}_{\overline{0}}$
are conjugate under inner automorphisms. So the Cartan subalgebras
of $\mathcal{G}$ are conjugate under inner automorphism.

\section{Realization of twisted Affine Lie superalgebras}

Let $\mathcal{G}$ be a basic simple Lie superalgebra with non degenerate invariant bilinear form $(.,.)$ and $\sigma$
an automorphism of finite order $m>1$. The eigenvalues of $\sigma$ are of the form $e^{\frac{2\pi k i}{m}}$,
$k\in\mathbb{Z}_m$ and hence admits the following $\mathbb{Z}_m$ grading:
 
\begin{equation}
 \mathcal{G}=\overset{m-1}{\underset{k=0}{\bigoplus}}\mathcal{G}_{k}, m\geq2
\end{equation}
such that 
\begin{equation}
 [\mathcal{G}_{i},\mathcal{G}_{j}]\subset\mathcal{G}_{i+j},\, i+j=i+j(\mbox{ mod } m)
\end{equation}
and
\begin{equation}
 \mathcal{G}_{k}= (\mathcal{G}_{k})_{\bar0}\oplus(\mathcal{G}_{k})_{\bar1}
\end{equation}
\begin{equation}
 \mathcal{G}_{k}=\{x\in \mathcal{G}|\sigma(x)=e^{\frac{2\pi k i}{m}.x}\}
\end{equation}
The twisted affine Lie superalgebra is defined to be 
\begin{equation}
\mathcal{G}^{(m)}=\left(\underset{k\in\mathbb{Z}_{m}}{\bigoplus}\mathbb{C}t^{k}\otimes\mathcal{G}_{k(\mbox{mod}m)}\right)\oplus\mathbb{C}c\oplus\mathbb{C}d
\end{equation}
The Lie superalgebra structure on $\mathcal{G}^{(m)}$ is such that $c$ is the canonical central element and
\begin{equation}
[x\otimes t^{m}+\lambda d, y\otimes t^{n}+\lambda_{1} d]=([x,y]\otimes t^{m+n}+\lambda ny\otimes t^{n}-\lambda_{1}mx\otimes t^{m}
+m\delta_{m,-n}(x,y)c
\end{equation}
where $x,y\in \mathcal{G}^{(m)}$ and $\lambda,\lambda_{1}\in \mathbb{C}$.
The element $d$ acts diagonally on $\mathcal{G}$ with interger eigenvalues and induces $\mathbb{Z}$ gradation.

\subsection{Cartan Involution}
Let $\mathfrak{g}$ is a compact Lie algebra if the group Int$\mathfrak{g}$ is compact.
An involution $\theta$ of a real semisimple Lie algebra $\mathfrak{g}_0$ such that symmetric bilinear form 
\begin{equation}
B_{\theta}(X,Y)=-B(X,\theta Y) 
\end{equation} is positive definite is called a Cartan involution.

\subsubsection{Cartan Involution of Contragradient Lie superalgebras}

$B$ is the supersymmeytic nondegenerate invariant bilinear form on $\mathcal{G}$
define \[ B_{\theta}(X,Y)=B(X,\theta Y)\]
We say that a real form of $\mathcal{G}$ has Cartan automorphism $\theta\in$aut$_{2,4}(\mathcal{G})$
if $B$ restricts to the Killing form on $\mathcal{G}_0$ and $B_\theta$ is symmetric negative definite on 
$\mathcal{G}^{(m)}$.

The bilinear form $(.,.)$ on $\mathcal{G}$ gives rise to a nondegenerate symmetric invariant form on $\mathcal{G}^{(m)}$ by
\begin{equation}
 B^{(m)}(\mathbb{C}[t,t^{-1}]\otimes\mathcal{G},\mathbb{C}K\oplus\mathbb{C}d)=0\end{equation}
\begin{equation}
\implies B^{(m)}(\underset{j\in\mathbb{Z}}{\bigoplus}t^{j}\otimes \mathcal{G}(\sigma)_{j mod m},\mathbb{C}K\oplus\mathbb{C}d)=0
\end{equation}
\begin{equation}
B^{(m)}(t^{j}\otimes X,t^{k}\otimes Y)=\lambda\delta^{j+k,0}B(X,Y)\end{equation}
\begin{equation}
 B^{(m)}(t^{j}\otimes X,K)=B^{(m)}(t^{j}\otimes X,d)=B^{(m)}(c,c)=B^{(m)}(d,d)=0\end{equation}
 \begin{equation}
  B^{m}(c,d)=1
 \end{equation}

\begin{proposition}

 Let $\theta\in$$aut$$_{2,4}(\mathcal{G}^{(m)})$. There exists a real form $\mathcal{G}^{(m)}_{\mathbb{R}}$
 such that $\theta$ restricts to a Cartan automorphism on $\mathcal{G}^{(m)}_{\mathbb{R}}$.
\end{proposition}

\begin{proof}

 Since $\theta$ is an $\mathcal{G}^{(m)}$ automorphism, it preserves $B$. namely
 $$B^{(m)}(X,Y)=B^{(m)}(\theta X,\theta Y)$$
 $B^{(m)}_{\theta}(X,Y)=B^{(m)}_{\theta}(Y,X)$, 
 $B^{(m)}_{\theta}(X,\theta X)=0$
$$B^{(m)}_{\theta}(X\otimes t^{m},Y\otimes t^{n})=B^{(m)}_{\theta}(Y\otimes t^{n},X\otimes t^{m})=$$
$$=t^{m+n}B(X,Y)$$
 for all $X,Y\in\mathcal{G}_{0}$
 $$B^{(m)}(c,X\otimes t^{k})=B(d,X\otimes t^{k})=B^{(m)}(d,d)=B^{(m)}(c,c)=0$$
 For $z\in L(t,t^{-1})\otimes\mathcal{G}_0$ and $X,Y\in L(t,t^{-1})\otimes\mathcal{G}_1$ 
 $$B^{(m)}_{\theta}(X,[Z,Y])=B^{(m)}(X,[\theta Z,\theta Y])=-B^{(m)}_{\theta}(X,[\theta Z,\theta Y])$$
 $$B^{(m)}_{\theta}(X,[Z,Y])=0$$ $\forall$ $X\in\mathbb{C}c$ or $\mathbb{C}d$  
 
 $\mathcal{G}_{\mathbb{R}}^{(m)}\simeq\mathcal{G}_{\overline{0}\mathbb{R}}^{(m)}\simeq\mathcal{G}_{\overline{0}\mathbb{R}}$.
The above three real forms are isomorphic. So the Cartan decomposition of $\mathcal{G}_{\mathbb{R}}^{(m)}$
are isomorphic

to $\mathcal{G}_{\overline{0}}$. 

$\mathcal{G}_{\overline{0}}=\mathfrak{k}_{0}\oplus\mathfrak{p}_{0}$

$B_{\theta}(X,[Z,Y])=\begin{cases}
\begin{array}{cc}
-B_{\theta}([Z,X],Y) & \mbox{if }Z\in\mathfrak{k}_{0}\\
B_{\theta}([Z,X],Y) & \mbox{if }Z\in\mathfrak{p}_{0}
\end{array}\end{cases}$

We say that a real form of $\mathcal{G}$ has Cartan automorphism $\theta\in$aut$_{2,4}(\mathcal{G})$
if $B$ restricts to the Killing form on $\mathcal{G}_0$ and $B_\theta$ is symmetric negative definite on 
$\mathcal{G}_{\mathbb{R}}$ and $B_{\theta}$ is symmetric bilinear form on $\mathcal{G}_{1} =\{1\otimes X_{1},1\otimes X_{2},\cdots,c,d\}$.
 $B_{\theta}(1\otimes X_{i},1\otimes X_{j})=\delta_{ij}$. It follows that $B_{\theta}$ negative definite on $\mathcal{G}_{\bar{1}\mathbb{R}}^{(m)}$.
 So it is concluded that $\theta$ is a 
 Cartan automorphism on $\mathcal{G}^{(m)}$.
\end{proof}

\section{Vogan diagram}
Let $\mathfrak{g}_0$ be a real semisimple Lie algebra, Let  $\mathfrak{g}$
be its complexification, let $\theta$ be a Cartan involution, let 
$\mathfrak{g}_{0}=\mathfrak{k}_{0}\oplus\mathfrak{p}_{0}$ be the corresponding Cartan decomposition
A maximally compact $\theta$ stable Cartan subalgebra $\mathfrak{h}_{0}=\mathfrak{k}_{0}\oplus\mathfrak{p}_{0}$
of $\mathfrak{g}_0$ with complexification $\mathfrak{h}=\mathfrak{k}\oplus\mathfrak{p}$ and we let 
$\triangle=\triangle(\mathfrak{g},\mathfrak{h})$ be the set of roots.
Choose a positive system $\triangle^{+}$  for $\triangle$ that takes $i\mathfrak{t}_0$ before $\mathfrak{a}$.
$\theta(\triangle^{+})=\triangle^{+}$\\
$\theta(\mathfrak{h}_{0})=\mathfrak{k}_{0}\oplus(-1)\mathfrak{p}_{0}$. 
Therefore $\theta$ permutes the simple roots. It must fix the simple roots that are imaginary and permute
in 2-cycles the simple roots that are complex.
By the Vogan diagram of the triple $(\mathfrak{g}_{0},\mathfrak{h}_{0},\triangle^{+}))$., we mean the Dynkin diagram of 
$\triangle^{+}$ with the 2 element orbits under $\theta$ so labeled and with the 1-element orbits painted or not, according as
the corresponding imaginary simple root is noncompact or compact.

\section{Twisted Affine Lie superalgebras}
A Dynkin diagram of $\mathcal{G}^{(m)}$ is obtained by adding a lowest weight (root)  to the Dynkin diagram of  $\mathcal{G}$.

\subsection{Root systems}
We have mentioned the lowest root because it has the relation with Kac-Dynkin label. We can get canonical nontrivial Kac-Dynkin
labels by lowest root from the fundamental representation. 

The root systems of twisted affine Lie superalgebra $OSp(2m|2n)^{(2)}$ is given by
 
$$\triangle=\{\frac{k}{2}-\delta_{1},\delta_{1}-\delta_{2},\cdots,\delta_{n-1}-\delta_{n},\delta_{n}-e_{1},
e_{1}-e_{2},\cdots,e_{m-1}-e_{m},e_{m}\}$$
The $\mathcal{G}_{0}$ representation $\mathcal{G}_{1}$ is the fundamental representation of 
$Osp(2m-1|2n)$ whose lowest weight is $-\delta_{1}$. For root systems of twisted affine Lie superalgebra $OSp(2|2n)^{(2)}$,
there exist an automorphism $\tau$ such that the invariant subsuperalgebra $\mathcal{G}_{0}$
is $OSp(1|2n)$. The simple root system of $\mathcal{G}_{0}$ is 
  
$$\triangle=\{\delta_{1}-\delta_{2},\cdots,\delta_{n-1}-\delta_{n},\delta_{n}\}$$
The lowest weight of the $\mathcal{G}_{1}$ representation of $\mathcal{G}_{0}$
is $\delta_{1}$. Similarly for twisted affine Lie superalgebra $Sl(1|2n+1)^{(4)}$, we  know the invariant subalgebra
can be taken to as $O(2n+1)$ and the lowest weight is $-\delta_{1}$.

\section{Vogan diagrams of affine Lie superalgebras}

Let $c$ the circling of vertices , $d$ diagram involution, $a_{s}$ numerical labeling and $D$ Dynkin diagram of $\mathcal{G}^{(m)}$.
 $S$ is defined to be the set of $d$ orbit vertices.\cite{Chuah:finite}
 \begin{definition}
 A  Vogan diagram $(c,d)$ on $D$  and one of the following holds:
 \begin{itemize}
  \item[({\it i})] $d$ fixes grey vertices
  \item [({\it ii})]  $\sum_{S}a_\alpha$ is odd.
 
 \end{itemize}

\end{definition}
 The $\gamma$, $\delta$ and  $c$ are expressed in terms of the bases given as follows

$\gamma=\overset{}{\underset{i=1}{\overset{n}{\sum}}a_{i}\alpha_{i}}$ , $\delta=\underset{i=0}{\overset{n}{\sum}}a_{i}\alpha_{i}$

Fix a set $\pi$ of simple roots of $\mathcal{G}$ , we take $\hat{\pi}=\{\alpha_{0}=\delta-\gamma\}\cup\pi$
be the simple roots of $\mathcal{G}^{(m)}$ ($\gamma$ is the highest weight in $\triangle_{\overline{0}}^
{(1)}\cup\triangle_{1}^{(1)}$).

 If $\theta$ extend to $aut_{2,4}$ (automorphism of order 2 or 4) then $\theta$ permutes the extreme weight spaces $\mathcal{G}^{(m)}$.
 Since $\theta|_{\mathcal{G}_{0}}$
 is represented by $(c,d)$ on $D_{0}$ (even part (set of even roots) of the Dynkin diagram),
 it permutes the simple root spaces of $\mathcal{G}_{0}$. Hence $\theta$ permutes the lowest weight spaces
 of $\mathcal{G}^{(m)}$ and $d$ extend to $inv(\mathcal{G}^{(m)})$ (where $inv$ is involution on (.)).

\begin{proposition}
 Let $\mathcal{G}_{\mathbb{R}}$ be a real form, with Cartan involution $\theta\in $inv$(\mathcal{G}_{\mathbb{R}})$ and Vogan diagram
 $(c,d)$ of $D_{0}$. The following are equivalent
 \begin{itemize}
 \item[({\it i})]  $\theta$ extend to aut$_{2,4}(\mathcal{G}^{(m)})$.
  \item [({\it ii})] $(\mathcal{G}_{\bar{0}\mathbb{R}})$ extend to a real form of $\mathcal{G}^{(m)}$.
  \item [({\it iii})] $(c,d)$ extend to a Vogan diagram on $D$
 \end{itemize}
\end{proposition}
\begin{proof}
  
$$S=$$\{vertices painted by p\}$$\cup$$\{white and adjacent 2-element d-orbits\}$$\cup$$\{grey and non adjacent 2-element d-orbits\}$$$$
 Let $D$ be the Dynkin diagram of $\mathcal{G}^{(m)})$ of simple root system
 $\Phi\cup\phi$($\Phi$ simple root system with $\phi$ lowest root) with $D=D_{\bar{0}}+D_{\bar{1}}$, where $D_{\bar {0}}$
 and $D_{\bar{1}}$ are respectively the white and grey vertices. The numerical label of the diagram shows $\sum_{\alpha\in D_{\bar{1}}}=2$ has either
 two grey vertices with label 1 or one grey vertex with label 2.
 
 \begin{itemize}
  \item[(i)] $D_{\bar{1}}=\{\gamma,\delta\}$ so the labelling of the odd vertices are 1.
  \item[(ii)] $D_{\bar{1}}=\{\gamma\}$ so labelling is 2 $( a_{\alpha}=2)$ on odd vertex.
 \end{itemize}
  $\theta\in $inv$(\mathcal{G}_{\mathbb{R}})$;  
 $\theta$ permutes the weightspaces $L(t,t^{-1})\otimes\mathcal{G}_{\bar{1}}$
 The rest part of proof of the proposition is followed the proof of the propostion 2.2 of \cite{chuah:vsuper}
 \end{proof}

When there is a $\sigma$ stable  compact Cartan subalgebra  then the 
Vogan diagrams are the following.

 The Vogan diagrams of   $\mathfrak{sl}(2m|2n)^{(2)}$ are
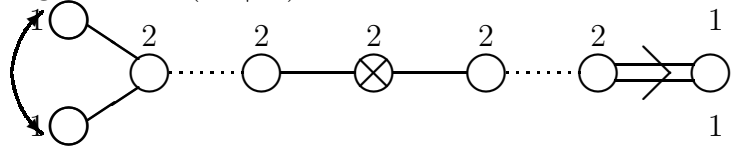
\begin{displaymath}
 \begin{picture}(80,20) \thicklines
     \put(-42,0){\circle{14}}   \put(0,0){\circle{14}} \put(42,0){\greydot} \put(84,0){\circle{14}} \put(126,0){\circle{14}}
  \put(-71,0){\vertexvn}
    \put(-35,0){\dottedline{4}(1,0)(28,0)}\put(7,0){\line(1,0){28}} \put(49,0){\line(1,0){28}}\put(91,0)
    {\dottedline{4}(1,0)(28,0)}
    \put(-71,0){\vertexvn}\put(126,0){\vertexbn}\put(168,0){\circle{14}}
             \put(-84,-20){\makebox(0,0){$1$}}
      \put(-84,20){\makebox(0,0){$1$}}
      \put(-42,14){\makebox(0,0){$2$}}
      \put(0,14){\makebox(0,0){$2$}} 
      \put(42,14){\makebox(0,0){$2$}} 
      \put(84,14){\makebox(0,0){$2$}}
       \put(126,14){\makebox(0,0){$2$}}
       \put(170,-20){\makebox(0,0){$1$}}
       \put(170,20){\makebox(0,0){$1$}}
         \qbezier(-93.5,0)(-92, -15)(-82, -23)
        \qbezier(-93.5,0)(-92, 15)(-82, 23)
\put(-82,23){\vector( 1, 1){0}}
\put(-82,-23){\vector( 1, -1){0}}
 \end{picture}
 \end{displaymath}
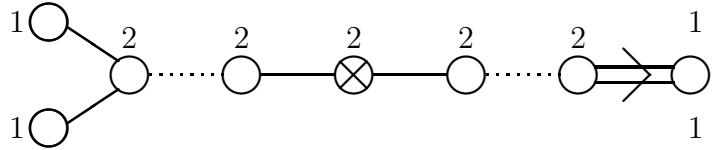
\vspace{1cm} 

\begin{displaymath}
 \begin{picture}(80,20) \thicklines
     \put(-42,0){\circle{14}}   \put(0,0){\circle{14}} \put(42,0){\greydot} \put(84,0){\circle{14}} \put(126,0){\circle{14}}
  \put(-71,0){\vertexvn}
    \put(-35,0){\dottedline{4}(1,0)(28,0)}\put(7,0){\line(1,0){28}} \put(49,0){\line(1,0){28}}\put(91,0)
    {\dottedline{4}(1,0)(28,0)}
    \put(-71,0){\vertexvn}\put(126,0){\vertexbn}\put(168,0){\circle{14}}
             \put(-84,-20){\makebox(0,0){$1$}}
      \put(-84,20){\makebox(0,0){$1$}}
      \put(-42,14){\makebox(0,0){$2$}}
      \put(0,14){\makebox(0,0){$2$}} 
      \put(42,14){\makebox(0,0){$2$}} 
      \put(84,14){\makebox(0,0){$2$}}
       \put(126,14){\makebox(0,0){$2$}}
       \put(170,-20){\makebox(0,0){$1$}}
       \put(170,20){\makebox(0,0){$1$}}
        
 \end{picture}
 \end{displaymath}
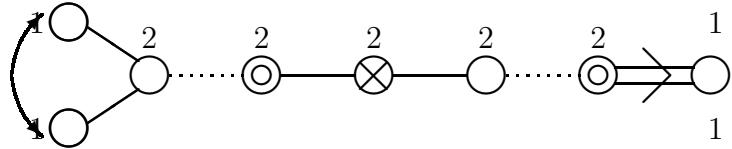
\vspace{1cm} 

\begin{displaymath}
 \begin{picture}(80,20) \thicklines
     \put(-42,0){\circle{14}}   \put(0,0){\circle{14}}\put(0,0){\circle{7}} \put(42,0){\greydot} \put(84,0){\circle{14}} \put(126,0){\circle{14}}\put(126,0){\circle{7}}
  \put(-71,0){\vertexvn}
    \put(-35,0){\dottedline{4}(1,0)(28,0)}\put(7,0){\line(1,0){28}} \put(49,0){\line(1,0){28}}\put(91,0)
    {\dottedline{4}(1,0)(28,0)}
    \put(-71,0){\vertexvn}\put(126,0){\vertexbn}\put(168,0){\circle{14}}
             \put(-84,-20){\makebox(0,0){$1$}}
      \put(-84,20){\makebox(0,0){$1$}}
      \put(-42,14){\makebox(0,0){$2$}}
      \put(0,14){\makebox(0,0){$2$}} 
      \put(42,14){\makebox(0,0){$2$}} 
      \put(84,14){\makebox(0,0){$2$}}
       \put(126,14){\makebox(0,0){$2$}}
       \put(170,-20){\makebox(0,0){$1$}}
       \put(170,20){\makebox(0,0){$1$}}
         \qbezier(-93.5,0)(-92, -15)(-82, -23)
        \qbezier(-93.5,0)(-92, 15)(-82, 23)
\put(-82,23){\vector( 1, 1){0}}
\put(-82,-23){\vector( 1, -1){0}}
 \end{picture}
 \end{displaymath}
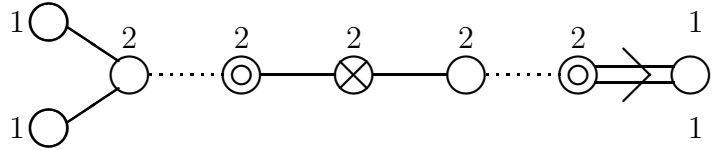
\vspace{1cm} 

\begin{displaymath}
 \begin{picture}(80,20) \thicklines
     \put(-42,0){\circle{14}}   \put(0,0){\circle{14}}\put(0,0){\circle{7}} \put(42,0){\greydot} \put(84,0){\circle{14}} \put(126,0){\circle{14}}\put(126,0){\circle{7}}
  \put(-71,0){\vertexvn}
    \put(-35,0){\dottedline{4}(1,0)(28,0)}\put(7,0){\line(1,0){28}} \put(49,0){\line(1,0){28}}\put(91,0)
    {\dottedline{4}(1,0)(28,0)}
    \put(-71,0){\vertexvn}\put(126,0){\vertexbn}\put(168,0){\circle{14}}
             \put(-84,-20){\makebox(0,0){$1$}}
      \put(-84,20){\makebox(0,0){$1$}}
      \put(-42,14){\makebox(0,0){$2$}}
      \put(0,14){\makebox(0,0){$2$}} 
      \put(42,14){\makebox(0,0){$2$}} 
      \put(84,14){\makebox(0,0){$2$}}
       \put(126,14){\makebox(0,0){$2$}}
       \put(170,-20){\makebox(0,0){$1$}}
       \put(170,20){\makebox(0,0){$1$}}
     \end{picture}
 \end{displaymath}
\vspace{1cm} 


 The Vogan diagrams of   $\mathfrak{sl}(2m|2n)^{(2)}$ are
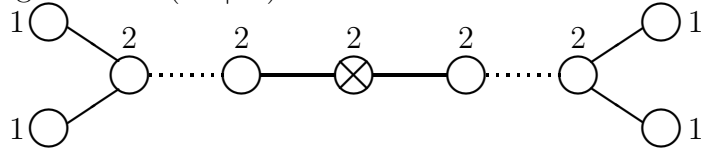
\begin{displaymath}
 \begin{picture}(80,20) \thicklines
     \put(-42,0){\circle{14}}   \put(0,0){\circle{14}} \put(42,0){\greydot} \put(84,0){\circle{14}} \put(126,0){\circle{14}}
  \put(-71,0){\vertexvn}
    \put(-35,0){\dottedline{4}(1,0)(28,0)}\put(7,0){\line(1,0){28}} \put(49,0){\line(1,0){28}}\put(91,0)
    {\dottedline{4}(1,0)(28,0)}
   \put(126,0){\vertexdn}
             \put(-84,-20){\makebox(0,0){$1$}}
      \put(-84,20){\makebox(0,0){$1$}}
      \put(-42,14){\makebox(0,0){$2$}}
      \put(0,14){\makebox(0,0){$2$}} 
      \put(42,14){\makebox(0,0){$2$}} 
      \put(84,14){\makebox(0,0){$2$}}
       \put(126,14){\makebox(0,0){$2$}}
       \put(170,-20){\makebox(0,0){$1$}}
       \put(170,20){\makebox(0,0){$1$}}
 \end{picture}
 \end{displaymath}
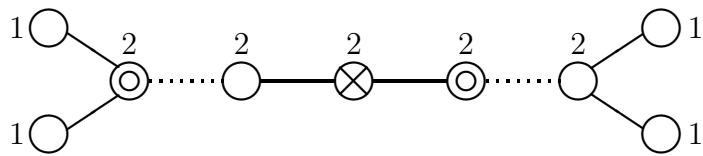
\vspace{1cm} 

\begin{displaymath}
 \begin{picture}(80,20) \thicklines
   \put(-42,0){\circle{7}}    \put(-42,0){\circle{14}}   \put(0,0){\circle{14}} \put(42,0){\greydot}
   \put(84,0){\circle{14}}  \put(84,0){\circle{7}} \put(126,0){\circle{14}}
 
 \put(-71,0){\vertexvn}
    \put(-35,0){\dottedline{4}(1,0)(28,0)}\put(7,0){\line(1,0){28}} \put(49,0){\line(1,0){28}}\put(91,0)
    {\dottedline{4}(1,0)(28,0)}
   \put(126,0){\vertexdn}
      
       \put(-84,-20){\makebox(0,0){$1$}}
      \put(-84,20){\makebox(0,0){$1$}}
      \put(-42,14){\makebox(0,0){$2$}}
      \put(0,14){\makebox(0,0){$2$}} 
      \put(42,14){\makebox(0,0){$2$}} 
      \put(84,14){\makebox(0,0){$2$}}
       \put(126,14){\makebox(0,0){$2$}}
       \put(170,-20){\makebox(0,0){$1$}}
       \put(170,20){\makebox(0,0){$1$}}
 \end{picture}
 \end{displaymath}
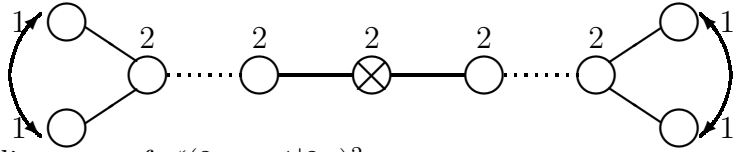
\vspace{1cm} 
\begin{displaymath}
 \begin{picture}(80,20) \thicklines
     \put(-42,0){\circle{14}}   \put(0,0){\circle{14}} \put(42,0){\greydot} \put(84,0){\circle{14}} \put(126,0){\circle{14}}
 
 \put(-71,0){\vertexvn}
    \put(-35,0){\dottedline{4}(1,0)(28,0)}\put(7,0){\line(1,0){28}} \put(49,0){\line(1,0){28}}\put(91,0)
    {\dottedline{4}(1,0)(28,0)}
   \put(126,0){\vertexdn}
      
      \qbezier(176.5, 0)(175, -15)(165, -23)
        \qbezier(176.5, 0)(175, 15)(165, 23)
\put(165,23){\vector( -1, 1){0}}
\put(165,-23){\vector( -1, -1){0}}
        \qbezier(-93.5,0)(-92, -15)(-82, -23)
        \qbezier(-93.5,0)(-92, 15)(-82, 23)
\put(-82,23){\vector( 1, 1){0}}
\put(-82,-23){\vector( 1, -1){0}}

       \put(-90,-20){\makebox(0,0){$1$}}
      \put(-90,20){\makebox(0,0){$1$}}
      \put(-42,14){\makebox(0,0){$2$}}
      \put(0,14){\makebox(0,0){$2$}} 
      \put(42,14){\makebox(0,0){$2$}} 
      \put(84,14){\makebox(0,0){$2$}}
       \put(126,14){\makebox(0,0){$2$}}
       \put(175,-20){\makebox(0,0){$1$}}
       \put(175,20){\makebox(0,0){$1$}}
 \end{picture}
 \end{displaymath}

 \bigskip 
\bigskip 

 The Vogan diagrams of   $\mathfrak{sl}(2m+1|2n)^{2}$ are
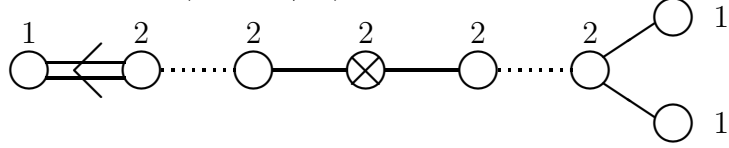
\begin{displaymath}
 \begin{picture}(80,20) \thicklines
   \put(-84,0){\circle{14}}   \put(-42,0){\circle{14}}   \put(0,0){\circle{14}} \put(42,0){\greydot}
   \put(84,0){\circle{14}} \put(126,0){\circle{14}}

    \put(-84,0){\vertexcn}\put(-35,0){\dottedline{4}(1,0)(28,0)}\put(7,0){\line(1,0){28}} \put(49,0){\line(1,0){28}}\put(91,0)
    {\dottedline{4}(1,0)(28,0)}
   \put(126,0){\vertexdn}

      \put(-84,14){\makebox(0,0){$1$}}
      \put(-42,14){\makebox(0,0){$2$}}
      \put(0,14){\makebox(0,0){$2$}} 
      \put(42,14){\makebox(0,0){$2$}} 
      \put(84,14){\makebox(0,0){$2$}}
       \put(126,14){\makebox(0,0){$2$}}
       \put(175,-20){\makebox(0,0){$1$}}
       \put(175,20){\makebox(0,0){$1$}}
 \end{picture}
 \end{displaymath}
 
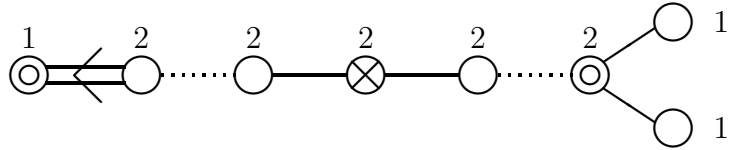
\vspace{1cm}
 \begin{displaymath}
 \begin{picture}(80,20) \thicklines
   \put(-84,0){\circle{14}}  \put(-84,0){\circle{7}}  \put(-42,0){\circle{14}}   \put(0,0){\circle{14}} \put(42,0){\greydot}
   \put(84,0){\circle{14}} \put(126,0){\circle{14}}\put(126,0){\circle{7}}

    \put(-84,0){\vertexcn}\put(-35,0){\dottedline{4}(1,0)(28,0)}\put(7,0){\line(1,0){28}} \put(49,0){\line(1,0){28}}\put(91,0)
    {\dottedline{4}(1,0)(28,0)}
   \put(126,0){\vertexdn}

      \put(-84,14){\makebox(0,0){$1$}}
      \put(-42,14){\makebox(0,0){$2$}}
      \put(0,14){\makebox(0,0){$2$}} 
      \put(42,14){\makebox(0,0){$2$}} 
      \put(84,14){\makebox(0,0){$2$}}
       \put(126,14){\makebox(0,0){$2$}}
       \put(175,-20){\makebox(0,0){$1$}}
       \put(175,20){\makebox(0,0){$1$}}
 \end{picture}
 \end{displaymath}
 
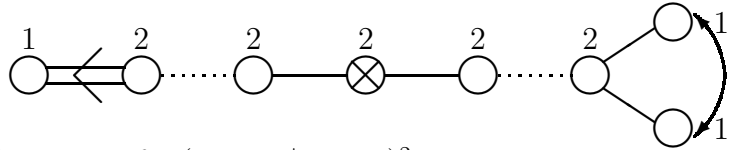
\vspace{1cm}
\begin{displaymath}
 \begin{picture}(80,20) \thicklines
   \put(-84,0){\circle{14}}   \put(-42,0){\circle{14}}   \put(0,0){\circle{14}} \put(42,0){\greydot}
   \put(84,0){\circle{14}} \put(126,0){\circle{14}}

    \put(-84,0){\vertexcn}\put(-35,0){\dottedline{4}(1,0)(28,0)}\put(7,0){\line(1,0){28}} \put(49,0){\line(1,0){28}}\put(91,0)
    {\dottedline{4}(1,0)(28,0)}
   \put(126,0){\vertexdn}
      
          \qbezier(176.5, 0)(175, -15)(165, -23)
        \qbezier(176.5, 0)(175, 15)(165, 23)
\put(165,23){\vector( -1, 1){0}}
\put(165,-23){\vector( -1, -1){0}}

      \put(-84,14){\makebox(0,0){$1$}}
      \put(-42,14){\makebox(0,0){$2$}}
      \put(0,14){\makebox(0,0){$2$}} 
      \put(42,14){\makebox(0,0){$2$}} 
      \put(84,14){\makebox(0,0){$2$}}
       \put(126,14){\makebox(0,0){$2$}}
       \put(175,-20){\makebox(0,0){$1$}}
       \put(175,20){\makebox(0,0){$1$}}
 \end{picture}
 \end{displaymath}

\bigskip 
\bigskip

 The Vogan diagrams of   $\mathfrak{sl}(2m+1|2n+1)^{2}$ are
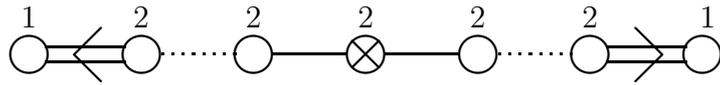
\begin{displaymath}
 \begin{picture}(80,20) \thicklines
   \put(-84,0){\circle{14}}   \put(-42,0){\circle{14}}   \put(0,0){\circle{14}} \put(42,0){\greydot} \put(84,0){\circle{14}}
   \put(126,0){\circle{14}}
 \put(168,0){\circle{14}} 
 
    \put(-84,0){\vertexcn}\put(-35,0){\dottedline{4}(1,0)(28,0)}\put(7,0){\line(1,0){28}} \put(49,0){\line(1,0){28}}\put(91,0)
    {\dottedline{4}(1,0)(28,0)}
  
      \put(126,0){\vertexbn}
      \put(-84,14){\makebox(0,0){$1$}}
      \put(-42,14){\makebox(0,0){$2$}}
      \put(0,14){\makebox(0,0){$2$}} 
      \put(42,14){\makebox(0,0){$2$}} 
      \put(84,14){\makebox(0,0){$2$}}
       \put(126,14){\makebox(0,0){$2$}}
       \put(170,14){\makebox(0,0){$1$}}
       
 \end{picture}
 \end{displaymath}
 
 
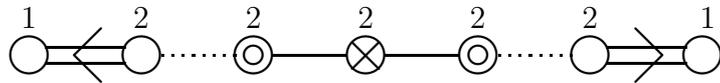
\vspace{1cm}
 
 \begin{displaymath}
 \begin{picture}(80,20) \thicklines
   \put(-84,0){\circle{14}}   \put(-42,0){\circle{14}}   \put(0,0){\circle{14}} \put(0,0){\circle{7}} \put(42,0){\greydot} \put(84,0){\circle{14}}\put(84,0){\circle{7}}
   \put(126,0){\circle{14}}
 \put(168,0){\circle{14}} 
 
    \put(-84,0){\vertexcn}\put(-35,0){\dottedline{4}(1,0)(28,0)}\put(7,0){\line(1,0){28}} \put(49,0){\line(1,0){28}}\put(91,0)
    {\dottedline{4}(1,0)(28,0)}
  
      \put(126,0){\vertexbn}
      \put(-84,14){\makebox(0,0){$1$}}
      \put(-42,14){\makebox(0,0){$2$}}
      \put(0,14){\makebox(0,0){$2$}} 
      \put(42,14){\makebox(0,0){$2$}} 
      \put(84,14){\makebox(0,0){$2$}}
       \put(126,14){\makebox(0,0){$2$}}
       \put(170,14){\makebox(0,0){$1$}}
       
 \end{picture}
 \end{displaymath}
 
 \vspace{1cm}
 
  The Vogan diagrams of  $\mathfrak{sl}(2|2n+1)^{(2)}$ are 
\begin{displaymath}
 \begin{picture}(80,20) \thicklines
     \put(-42,0){\circle{14}}   \put(0,0){\circle{14}} \put(42,0){\greydot} \put(84,0){\circle{14}} \put(126,0){\circle{14}}
 
 \put(-71,0){\vertexvn}\put(126,0){\vertexbn}\put(168,0){\circle{14}}
    \put(-35,0){\dottedline{4}(1,0)(28,0)}\put(7,0){\line(1,0){28}} \put(49,0){\line(1,0){28}}\put(91,0)
    {\dottedline{4}(1,0)(28,0)}
   
      \put(-76,-15){\line(0,1){30}}\put(-68,-15){\line(0,1){30}}
      
              \put(-84,-20){\makebox(0,0){$1$}}
      \put(-84,20){\makebox(0,0){$1$}}
      \put(-42,14){\makebox(0,0){$2$}}
      \put(0,14){\makebox(0,0){$2$}} 
      \put(42,14){\makebox(0,0){$2$}} 
      \put(84,14){\makebox(0,0){$2$}}
       \put(126,14){\makebox(0,0){$2$}}
       \put(170,-20){\makebox(0,0){$$}}
       \put(170,14){\makebox(0,0){$1$}}
 \end{picture}
 \end{displaymath}
 \vspace{1cm}
 
 \begin{displaymath}
 \begin{picture}(80,20) \thicklines
     \put(-42,0){\circle{14}}   \put(0,0){\circle{14}} \put(42,0){\greydot} \put(84,0){\circle{14}} \put(126,0){\circle{14}}
 
 \put(-71,0){\vertexvn}\put(126,0){\vertexbn}\put(168,0){\circle{14}}
    \put(-35,0){\dottedline{4}(1,0)(28,0)}\put(7,0){\line(1,0){28}} \put(49,0){\line(1,0){28}}\put(91,0)
    {\dottedline{4}(1,0)(28,0)}
   
      \put(-76,-15){\line(0,1){30}}\put(-68,-15){\line(0,1){30}}
      
              \put(-84,-20){\makebox(0,0){$1$}}
      \put(-84,20){\makebox(0,0){$1$}}
      \put(-42,14){\makebox(0,0){$2$}}
      \put(0,14){\makebox(0,0){$2$}} 
      \put(42,14){\makebox(0,0){$2$}} 
      \put(84,14){\makebox(0,0){$2$}}
       \put(126,14){\makebox(0,0){$2$}}
       \put(170,-20){\makebox(0,0){$$}}
       \put(170,14){\makebox(0,0){$1$}}
       \qbezier(-93.5,0)(-92, -15)(-82, -23)
        \qbezier(-93.5,0)(-92, 15)(-82, 23)
\put(-82,23){\vector( 1, 1){0}}
\put(-82,-23){\vector( 1, -1){0}}
 \end{picture}
 \end{displaymath}
 \vspace{1cm}
 
 \begin{displaymath}
 \begin{picture}(80,20) \thicklines
     \put(-42,0){\circle{14}}   \put(0,0){\circle{14}} \put(0,0){\circle{7}}\put(42,0){\greydot} \put(84,0){\circle{14}} \put(84,0){\circle{7}}\put(126,0){\circle{14}}
 
 \put(-71,0){\vertexvn}\put(126,0){\vertexbn}\put(168,0){\circle{14}}
    \put(-35,0){\dottedline{4}(1,0)(28,0)}\put(7,0){\line(1,0){28}} \put(49,0){\line(1,0){28}}\put(91,0)
    {\dottedline{4}(1,0)(28,0)}
   
      \put(-76,-15){\line(0,1){30}}\put(-68,-15){\line(0,1){30}}
      
              \put(-84,-20){\makebox(0,0){$1$}}
      \put(-84,20){\makebox(0,0){$1$}}
      \put(-42,14){\makebox(0,0){$2$}}
      \put(0,14){\makebox(0,0){$2$}} 
      \put(42,14){\makebox(0,0){$2$}} 
      \put(84,14){\makebox(0,0){$2$}}
       \put(126,14){\makebox(0,0){$2$}}
       \put(170,-20){\makebox(0,0){$$}}
       \put(170,14){\makebox(0,0){$1$}}
 \end{picture}
 \end{displaymath}
 \vspace{1cm}
 
  \begin{displaymath}
 \begin{picture}(80,20) \thicklines
     \put(-42,0){\circle{14}}   \put(0,0){\circle{14}} \put(0,0){\circle{7}}\put(42,0){\greydot} \put(84,0){\circle{14}} \put(84,0){\circle{7}}\put(126,0){\circle{14}}
 
 \put(-71,0){\vertexvn}\put(126,0){\vertexbn}\put(168,0){\circle{14}}
    \put(-35,0){\dottedline{4}(1,0)(28,0)}\put(7,0){\line(1,0){28}} \put(49,0){\line(1,0){28}}\put(91,0)
    {\dottedline{4}(1,0)(28,0)}
   
      \put(-76,-15){\line(0,1){30}}\put(-68,-15){\line(0,1){30}}
      
              \put(-84,-20){\makebox(0,0){$1$}}
      \put(-84,20){\makebox(0,0){$1$}}
      \put(-42,14){\makebox(0,0){$2$}}
      \put(0,14){\makebox(0,0){$2$}} 
      \put(42,14){\makebox(0,0){$2$}} 
      \put(84,14){\makebox(0,0){$2$}}
       \put(126,14){\makebox(0,0){$2$}}
       \put(170,-20){\makebox(0,0){$$}}
       \put(170,14){\makebox(0,0){$1$}}
       \qbezier(-93.5,0)(-92, -15)(-82, -23)
        \qbezier(-93.5,0)(-92, 15)(-82, 23)
\put(-82,23){\vector( 1, 1){0}}
\put(-82,-23){\vector( 1, -1){0}}
 \end{picture}
 \end{displaymath}
 \vspace{1cm}
 
\vspace{1cm}
 The Vogan diagrams of  $\mathfrak{sl}(2|2n)^{(2)}$ are 
\begin{displaymath}
 \begin{picture}(80,20) \thicklines
     \put(-42,0){\circle{14}}   \put(0,0){\circle{14}} \put(42,0){\greydot} \put(84,0){\circle{14}} \put(126,0){\circle{14}}
 
 \put(-71,0){\vertexvn}
    \put(-35,0){\dottedline{4}(1,0)(28,0)}\put(7,0){\line(1,0){28}} \put(49,0){\line(1,0){28}}\put(91,0)
    {\dottedline{4}(1,0)(28,0)}
   \put(126,0){\vertexdn}
      \put(-76,-15){\line(0,1){30}}\put(-68,-15){\line(0,1){30}}
      
              \put(-84,-20){\makebox(0,0){$1$}}
      \put(-84,20){\makebox(0,0){$1$}}
      \put(-42,14){\makebox(0,0){$2$}}
      \put(0,14){\makebox(0,0){$2$}} 
      \put(42,14){\makebox(0,0){$2$}} 
      \put(84,14){\makebox(0,0){$2$}}
       \put(126,14){\makebox(0,0){$2$}}
       \put(170,-20){\makebox(0,0){$1$}}
       \put(170,20){\makebox(0,0){$1$}}
 \end{picture}
 \end{displaymath}
 \vspace{1cm}

\begin{displaymath}
 \begin{picture}(80,20) \thicklines
     \put(-42,0){\circle{14}}   \put(0,0){\circle{14}} \put(42,0){\greydot} \put(84,0){\circle{14}} \put(126,0){\circle{14}}
 
 \put(-71,0){\vertexvn}
    \put(-35,0){\dottedline{4}(1,0)(28,0)}\put(7,0){\line(1,0){28}} \put(49,0){\line(1,0){28}}\put(91,0)
    {\dottedline{4}(1,0)(28,0)}
   \put(126,0){\vertexdn}
      \put(-76,-15){\line(0,1){30}}\put(-68,-15){\line(0,1){30}}
      
        \qbezier(176.5, 0)(175, -15)(165, -23)
        \qbezier(176.5, 0)(175, 15)(165, 23)
\put(165,23){\vector( -1, 1){0}}
\put(165,-23){\vector( -1, -1){0}}
        \qbezier(-93.5,0)(-92, -15)(-82, -23)
        \qbezier(-93.5,0)(-92, 15)(-82, 23)
\put(-82,23){\vector( 1, 1){0}}
\put(-82,-23){\vector( 1, -1){0}}
       \put(-84,-20){\makebox(0,0){$1$}}
      \put(-84,20){\makebox(0,0){$1$}}
      \put(-42,14){\makebox(0,0){$2$}}
      \put(0,14){\makebox(0,0){$2$}} 
      \put(42,14){\makebox(0,0){$2$}} 
      \put(84,14){\makebox(0,0){$2$}}
       \put(126,14){\makebox(0,0){$2$}}
       \put(170,-20){\makebox(0,0){$1$}}
       \put(170,20){\makebox(0,0){$1$}}
 \end{picture}
 \end{displaymath}
 \vspace{2cm}
 
The Vogan diagrams of   $\mathfrak{osp}(2m|2n)^{(2)}$ are
\begin{displaymath}
 \begin{picture}(80,20) \thicklines
   \put(-84,0){\circle{14}}   \put(-42,0){\circle{14}}   \put(0,0){\circle{14}} \put(42,0){\greydot} \put(84,0){\circle{14}}
   \put(126,0){\circle{14}}
 \put(168,0){\circle{14}} 
 
    \put(-84,0){\vertexcn}\put(-35,0){\dottedline{4}(1,0)(28,0)}\put(7,0){\line(1,0){28}} \put(49,0){\line(1,0){28}}\put(91,0)
    {\dottedline{4}(1,0)(28,0)}
  
      \put(126,0){\vertexbn}
      \put(-84,14){\makebox(0,0){$1$}}
      \put(-42,14){\makebox(0,0){$1$}}
      \put(0,14){\makebox(0,0){$1$}} 
      \put(42,14){\makebox(0,0){$1$}} 
      \put(84,14){\makebox(0,0){$1$}}
       \put(126,14){\makebox(0,0){$1$}}
       \put(170,14){\makebox(0,0){$1$}}
       
 \end{picture}
 \end{displaymath}
 
\vspace{1cm}

 \hspace{7cm} \vdots
 
 \begin{displaymath}
 \begin{picture}(80,20) \thicklines
   \put(-84,0){\circle{14}}   \put(-42,0){\circle{14}}   \put(0,0){\circle{14}} \put(0,0){\circle{7}}\put(42,0){\greydot} \put(84,0){\circle{14}}
   \put(126,0){\circle{14}} \put(126,0){\circle{7}}
 \put(168,0){\circle{14}} 
 
    \put(-84,0){\vertexcn}\put(-35,0){\dottedline{4}(1,0)(28,0)}\put(7,0){\line(1,0){28}} \put(49,0){\line(1,0){28}}\put(91,0)
    {\dottedline{4}(1,0)(28,0)}
  
      \put(126,0){\vertexbn}
      \put(-84,14){\makebox(0,0){$1$}}
      \put(-42,14){\makebox(0,0){$1$}}
      \put(0,14){\makebox(0,0){$1$}} 
      \put(42,14){\makebox(0,0){$1$}} 
      \put(84,14){\makebox(0,0){$1$}}
       \put(126,14){\makebox(0,0){$1$}}
       \put(170,14){\makebox(0,0){$1$}}
       
 \end{picture}
 \end{displaymath}
 
\vspace{1cm}

 The lowest weight representation $\mathcal{G}_{1}$ of  $\mathcal{G}_{0}$  is $-\delta_{1}$ and that makes
 the following Dynkin diagram for  $\mathfrak{osp}(2|2n)^{2}$.
  The Vogan diagrams of   $\mathfrak{osp}(2|2n)^{(2)}$ are
\begin{displaymath}
 \begin{picture}(80,20) \thicklines
   \put(-84,0){\circle*{14}}   \put(-42,0){\circle{14}}   \put(0,0){\circle{14}} \put(38,0){\dottedline{4}(1,0)(28,0)}
   \put(84,0){\circle{14}}
   \put(126,0){\circle{14}}
 \put(168,0){\circle*{14}} 
 
    \put(-84,0){\vertexcn}\put(-35,0){\dottedline{4}(1,0)(28,0)}\put(7,0){\line(1,0){28}} \put(49,0){\line(1,0){28}}\put(91,0)
    {\dottedline{4}(1,0)(28,0)}
  
      \put(126,0){\vertexbn}
      \put(-84,14){\makebox(0,0){$1$}}
      \put(-42,14){\makebox(0,0){$2$}}
      \put(0,14){\makebox(0,0){$2$}} 
      \put(42,14){\makebox(0,0){$$}} 
      \put(84,14){\makebox(0,0){$2$}}
       \put(126,14){\makebox(0,0){$2$}}
       \put(170,14){\makebox(0,0){$2$}}
       
 \end{picture}
 \end{displaymath}
 
\vspace{1cm}

  \hspace{7cm} \vdots

 \begin{displaymath}
 \begin{picture}(80,20) \thicklines
   \put(-84,0){\circle*{14}}   \put(-42,0){\circle{14}}   \put(0,0){\circle{14}}\put(0,0){\circle{7}} \put(38,0){\dottedline{4}(1,0)(28,0)}
   \put(84,0){\circle{14}}
   \put(126,0){\circle{14}}
 \put(168,0){\circle*{14}} 
 
    \put(-84,0){\vertexcn}\put(-35,0){\dottedline{4}(1,0)(28,0)}\put(7,0){\line(1,0){28}} \put(49,0){\line(1,0){28}}\put(91,0)
    {\dottedline{4}(1,0)(28,0)}
  
      \put(126,0){\vertexbn}
      \put(-84,14){\makebox(0,0){$1$}}
      \put(-42,14){\makebox(0,0){$2$}}
      \put(0,14){\makebox(0,0){$2$}} 
      \put(42,14){\makebox(0,0){$$}} 
      \put(84,14){\makebox(0,0){$2$}}
       \put(126,14){\makebox(0,0){$2$}}
       \put(170,14){\makebox(0,0){$2$}}
       
 \end{picture}
 \end{displaymath}
 
\vspace{1cm}
 
  The lowest weight representation $\mathcal{G}_{1}$ of  $\mathcal{G}_{0}$ is $-\delta_{1}$ and that makes
 the following Dynkin diagram for  $\mathfrak{sl}(1|2n+1)^{(4)}$.
  The Vogan diagrams of   $\mathfrak{sl}(1|2n+1)^{4}$ are
\begin{displaymath}
 \begin{picture}(80,20) \thicklines
   \put(-84,0){\circle{14}}   \put(-42,0){\circle{14}}   \put(0,0){\circle{14}} \put(38,0){\dottedline{4}(1,0)(28,0)}
   \put(84,0){\circle{14}}
   \put(126,0){\circle{14}}
 \put(168,0){\circle*{14}} 
 
    \put(-84,0){\vertexcn}\put(-35,0){\dottedline{4}(1,0)(28,0)}\put(7,0){\line(1,0){28}} \put(49,0){\line(1,0){28}}\put(91,0)
    {\dottedline{4}(1,0)(28,0)}
  
      \put(126,0){\vertexbn}
      \put(-84,14){\makebox(0,0){$1$}}
      \put(-42,14){\makebox(0,0){$1$}}
      \put(0,14){\makebox(0,0){$1$}} 
      \put(42,14){\makebox(0,0){$$}} 
      \put(84,14){\makebox(0,0){$1$}}
       \put(126,14){\makebox(0,0){$1$}}
       \put(170,14){\makebox(0,0){$1$}}
       
 \end{picture}
 \end{displaymath}
 
\vspace{1cm}
 
  \hspace{7cm} \vdots
 
 \begin{displaymath}
 \begin{picture}(80,20) \thicklines
   \put(-84,0){\circle{14}}   \put(-42,0){\circle{14}}   \put(0,0){\circle{14}} \put(38,0){\dottedline{4}(1,0)(28,0)}
   \put(84,0){\circle{14}}\put(84,0){\circle{7}}
   \put(126,0){\circle{14}}
 \put(168,0){\circle*{14}} 
 
    \put(-84,0){\vertexcn}\put(-35,0){\dottedline{4}(1,0)(28,0)}\put(7,0){\line(1,0){28}} \put(49,0){\line(1,0){28}}\put(91,0)
    {\dottedline{4}(1,0)(28,0)}
  
      \put(126,0){\vertexbn}
      \put(-84,14){\makebox(0,0){$1$}}
      \put(-42,14){\makebox(0,0){$1$}}
      \put(0,14){\makebox(0,0){$1$}} 
      \put(42,14){\makebox(0,0){$$}} 
      \put(84,14){\makebox(0,0){$1$}}
       \put(126,14){\makebox(0,0){$1$}}
       \put(170,14){\makebox(0,0){$1$}}
       
 \end{picture}
 \end{displaymath}

\end{document}